\documentclass{llncs}
\usepackage[ruled]{algorithm2e} 
\usepackage{amsmath,amsfonts}

\SetAlFnt{\small}
\SetAlCapFnt{\small}
\SetAlCapNameFnt{\small}
\SetAlCapHSkip{0pt}
\IncMargin{-\parindent}

\newcommand{\cc}{\mbox{\boldmath $c$}}
\newcommand{\argmax}{\mbox{\rm argmax}}
\begin{document}
\title{Response Prediction for Low-Regret Agents}
\author{Saeed Alaei \and 
Ashwinkumar Badanidiyuru \and \\
Mohammad Mahdian \and 
Sadra Yazdanbod}
\institute{Google}

\maketitle

\begin{abstract}
Companies like Google and Microsoft run billions of auctions every day
to sell advertising opportunities. Any change to the rules of these
auctions can have a tremendous effect on the revenue of the company
and the welfare of the advertisers and the users. Therefore, any
change requires careful evaluation of its potential impacts.
Currently, such impacts are often evaluated by running simulations or
small controlled experiments. This, however, misses the important
factor that the advertisers respond to changes. Our goal is to build a
theoretical framework for predicting the actions of an agent (the
advertiser) that is optimizing her actions in an uncertain
environment. We model this problem using a variant of the multi-armed
bandit setting where playing an arm is costly. The cost of each arm
changes over time and is publicly observable. The value of playing an
arm is drawn stochastically from a static distribution and is observed
by the agent and not by us. We, however, observe the actions of the
agent. Our main result is that assuming the agent is playing a
strategy with a regret of at most $f(T)$ within the first $T$ rounds,
we can learn to play the multi-armed bandits game (without observing
the rewards) in such a way that the regret of our selected actions is
at most $O(k^4(f(T)+1)\log(T))$, where $k$ is the number of arms.
\end{abstract}

\section{Introduction}

Over the last two decades, the online advertising market has emerged
as one of the most important application areas of auctions. Companies
like Google and Microsoft run billions of auctions every day to sell
advertising opportunities worth hundreds of millions of dollars. Rules
of these auctions have undergone frequent change, often prompted by
the release of new features (such as ads with additional site links or
ads taking advantage of re-targeting lists) or by optimizations in the
auction system (such as a new reserve price algorithm or a new
algorithm for estimating click probabilities).  Any such change can
have tremendous impact on the revenue of the company and the welfare
of the advertisers and the users. Therefore, any proposed change to
the auction system goes through a rigorous vetting process to evaluate
its potential impacts and decide, based on the results of the
evaluation and current business priorities, whether the proposal
merits a launch.

Currently, the main tools used for evaluating a proposed launch is
running simulations~\cite{mizuta2000agent} or small controlled
experiments~\cite{tang2010overlapping}. These approaches, however,
miss the important factor that the advertisers respond to
changes. This is evident in the case of simulations, where the bids
advertisers have submitted for the existing auction are used to
simulate the new proposed auction. In the case of controlled
experiments, the trouble is that the treatment often has to be applied
to all or none of advertisers in an auction. This, together with the
fact that advertisers overlap imperfectly on the set of auctions they
participate in, makes it practically impossible to select random
treatment and control groups of advertisers, treat all of the auctions
the treatment set of advertisers participate in while leaving all
auctions that the control group participate in untreated
(See~\cite{backstrom2011network} for a discussion of a very similar
problem in the context of social
networks)\footnote{See~\cite{lang2007finding} for an attempt to solve
  this problem by restricting the experiment to small micro-markets.
  Note that this has the obvious disadvantage of biasing the
  experiment toward a non-representative set of advertisers and
  auctions.}. In practice, experiments are run with a random set of
auctions (typically 1\% or less of all auctions) as the treatment
group. This means that for each advertiser only a very small
percentage of their auctions is treated, leading to a treatment effect
that is well smaller than the noise in the system, and is hence
practically unobservable by the
advertiser.\footnote{See~\cite{chawla2014mechanism} for an interesting
  theoretical treatment of this setting. It turns out that assuming
  that the advertisers are fully rational and react even to a small
  change in the auction, even treating a small percentage of each
  advertiser's auctions is enough to extrapolate their response to a
  full treatment. In practice, however, there is too much noise and
  fluctuation in the system for advertisers to be able to observe and
  respond to a change that, for example, increases their cost per
  click by 10\% in 1\% of their auctions.}

In this paper, our goal is to build a theoretical framework for
predicting advertiser response based on observations about their past
actions. Our model is driven by a few important considerations. First,
the advertisers face an uncertain environment, and optimize their
objective in presence of uncertainty. As in~\cite{NST15}, we capture
this by modeling the advertiser as an agent solving a regret
minimization problem in a multi-armed bandit setting. In our
motivating application, each arm can correspond to an ad slot the agent
can purchase or to a discretized value of the bid the agent
submits. We make no assumption on the type of algorithm the agent is
using except that it has bounded regret. Second, we are concerned with
an environment that is changing, and therefore requires the agent to
respond to this change. We model this by assuming each arm has a cost,
and in each round, the agent is informed about the cost before he has
to choose which arm to play. This is the main point of difference
between our model and the model in~\cite{NST15}, and is an important
element of our model, since without this, to predict which arm an
agent is going to play, it is enough to look at their past history and
select the arm that is played most often. The assumption that the cost
of each arm is observed before the agent picks which arm to play is
not entirely accurate in our motivating application, since advertisers
only learn about the cost of their ad after it is placed. However,
given that in practice costs change continuously over time, the
advertisers can use the cost of each arm in the recent past as a proxy
for its current cost. Therefore, we feel this assumption is a
justified approximation of the real scenario. 

Finally, we model the objective of our prediction problem. In our
model, once the agent decides which arm to play, they receive a reward
from that arm that is drawn stochastically from a static
distribution.\footnote{In our motivating application, the reward can be
  the profit the advertiser makes if the user clicks on their ad and
  makes a purchase, or zero otherwise. In this case, the assumption
  that the reward distribution is static means that the profit per
  conversion and the conversion probability are fixed over time. This
  is not entirely accurate, but is a reasonable approximation of the
  reality, since while these parameters change over time, they tend to
  change at a slow pace.}  This reward is observed by the agent but
not by us. All we observe is the cost of the arms and the arm that the
agent plays. Over time, we would like to be able to ``predict'' which
arm the agent plays.  We need to be careful about the way we capture
this in our model. For example, if two of the arms always have the
same cost and the same reward, the agent's choice between them is
arbitrary and can never be predicted. Also, if an arm has never been
played (e.g., since its cost has been infinity so far), we cannot be
expected to predict the first time it is played. For these reasons, we
evaluate our prediction algorithm by the regret of its actions. Our
main result is an algorithm that by observing the actions of the agent
learns to play the multi-armed bandit problem with a regret that is
close to that of the agent.
Furthermore, we show if the optimal arm, i.e., the arm with highest reward and lowest cost, is unique at every step, the number of predictions of our algorithm that is not exactly the same as the agent actions is upper bounded. Our upper bound depends on the distance  between the optimal arm and the second optimal arm at every step.
%
%

Since we evaluate our algorithm by the regret of its actions, it can be seen as a regret minimization algorithm which is a very well studied subject. The distinguishing point of between our work and previous work in regret minimization is that in our setting the algorithm does not observe the payoffs (not even the payoff of the arm it selects) which is the essential input for regret minimization algorithms in the literature \cite{BC12}.

\section{Related Work}
The closest previous work to this paper is \cite{NST15}, where the authors study a model for learning an agent's valuations based on the agent's responses. Similar to this paper,~\cite{NST15} does not assume that the agent always chooses a myopically optimal action, but assumes that the agent chooses its actions using a no-regret learning algorithm. There are two main differences between the model in~\cite{NST15} and in our paper. The first difference is that \cite{NST15} studies a single parameter setting where each agent reports a single bid, whereas we study a multi-parameter setting where the agent can pick one of many actions and the utility of each action might not be related to the others. Hence as a model one can reduce \cite{NST15} to our model by disretization. Another key difference between the two papers is the metric. The goal of \cite{NST15} is to study sample complexity of computing a set whose Hausdoff distance from the ``rationalizable set'' of valuations is not large. In the current paper the metric is regret of the algorithm with respect to the agent's valuation.
Another related work is~\cite{FKT10}, where the authors study the problem of mimicking an opponent in a 2 player gaming setting when we cannot observe the payoff and the only thing that is observable is the action of the opponent. 

As we discussed in the introduction, our results can be used for bid prediction if the arms correspond to discretized values for the bids the agents submits.
There are a number of papers \cite{XG13,CDE07,PK11,BG11} on this subject that model different objectives and behaviors of the agents. 
However, most of them rely on an estimation of the agent's private values so they can be used for bid prediction.
Also, most of these papers ignore the fact that the agents often faces an uncertain environment that they learn over time, and the optimizations happen in presence of uncertainty. 

Another line of related work is on designing mechanisms for agents that follow no-regret strategies. For example~\cite{Braverman} studies an auction design problem in such a model.

Outside of computer science there is also a rich literature in Economics studying inference in auctions under equilibrium assumptions. A survey of this literature can be found in \cite{AH07}. This approach has been used to study a wide variety of settings such as arbitrary normal form games \cite{KS15}, static first-price auctions \cite{GPV00}, extension to risk-verse bidders \cite{GPV09,CGPV03}, sequential auctions \cite{BP03} and sponsored search auctions \cite{V07,SD10}. 


\section{Model}
In this section we describe our theoretical framework for predicting advertiser response based on observations about their past actions.
In our model, an agent (representing an advertiser in our motivating
application) plays a multi-armed bandit game with $k$ arms. In each of the
time steps $t=1, 2, \ldots$, each arm $i$ has a cost $c^t_i$. These
costs can be different in each time step, but they are observed by the agent and by us at the beginning of each time step. The reward (also called the value) of
playing arm $i$ in any time step is drawn from a distribution
$\mathcal{D}_i$ with expected value $0 \leq v_i\leq 1$. The agent does not know
$\mathcal{D}_i$ or $v_i$, but after playing an arm, privately observes its reward. In our motivating application, each arm can correspond to a bid value the advertiser can submit. The reward of an arm is the value the advertiser receives (e.g., by selling a product through the click-through on their ad), and the cost corresponds to the amount they have to pay for their ad. In this context, the assumptions that the costs are observed by the advertiser as well as the auctioneer, that the distribution $\mathcal{D}_i$ is unknown, and that the reward is observed by the advertiser but not by the auctioneer all make sense.

As the costs are different at each time step, the optimal action ~
$o_t = \argmax_{i\in [k]} \{v_i-c^t_i\}$ for the agent can also be
different. Since the agent does not know $v_i$'s, she might play an
arm that is not necessarily optimal.  Let $a_t$ be the arm that the
agent picks at step $t$. As a result of this choice, the agent accrues
a regret of $ar_t=(v_{o_t}-c^t_{o_t})-(v_{a_t}-c^t_{a_t})$ at time
step $t$. We assume that the agent uses an arbitrary bounded-regret
strategy, i.e., her total regret $\sum_{t=1}^T ar_t$ up to time $T$ is 
bounded by a function $f(T)$ for each time step $T$.

The goal is to design an algorithm that in each time step $t$, given
the history of the agent actions up to this time step (i.e., the costs $\cc^1,
\ldots, \cc^{t-1}$ and the actions $a_1, \ldots, a_{t-1}$ of the
agent, but not the rewards the agent has received) and the costs
$\cc^t$ of the arms in this time step, picks an arm $p_t$. Because of
this choice, the algorithm accrues a regret of
$pr_t=(v_{o_t}-c^t_{o_t})-(v_{p_t}-c^t_{p_t})$ at step $t$. Our metric
for the algorithm's performance is measured by the total regret it
achieves as compared to the regret of the agent.

Our main result is that 
there exists an algorithm with a regret bound of $O(k^4 (f(T)+1) \log(T))$.

\section{Prediction Algorithm}
In this section, we describe our prediction algorithm.
A key step in designing the algorithm is our assumption that the agent's regret is bounded by $f(t)$ for each time step $t$. This allows us to define a set of values for the agent that are consistent with their actions so far and their regret bound. A value vector $v$ is consistent with the actions up to time $t$ if there exists a regret vector $r$ such that:
\begin{equation}
\label{eqn:cv}
\begin{array}{ll}
 v_{a_\ell}-c^\ell_{a_\ell} \geq v_{i}-c^\ell_i-r_\ell & \forall \ell \in [t-1],\forall i\in [k]\\
	     \sum_{j\leq \ell} r_j \leq f(\ell) & \forall \ell \in [t-1]
\end{array}
\end{equation}
We denote the set of consistent values at time $t$ with $CV(t)$. 
Note that for every $v\in CV(t)$, the optimal arm is $\argmax_i\{ v_i -c^t_i\}$.
The main idea of the algorithm is to pick an arm which is the optimal arm for the largest portion of $CV(t)$. 
Formally, for each arm $i$ define $w_i$ as the probability that $i$ is the optimal arm for a vector $v\in CV(t)$ chosen uniformly at random. 
At every time step $t$, our algorithm picks the arm $i$ with the highest $w_i$.

\begin{algorithm}[h]
\SetAlgoNoLine
 $CV(0) = \{v\ |\  0\leq v_i\leq 1, \forall i\}$\;
 \For{ each time step $t$}{
     $c^t \leftarrow \textrm{costs of playing arms at time step $t$}$\;
     $CV(t) \gets  \textrm{the set of consistent values at time step $t$}$ \;
     $w_i:=Pr_{v\sim Unif(CV(t))} [ v_{i}-c^t_i\geq v_{j}-c^t_{j},\ \forall j]$\;
     $p_t\gets \arg\max_{i}\ w_i$\;
 }

%
\caption{Prediction Algorithm}\label{main-alg}
\end{algorithm}
The time complexity of our algorithm at each time step is equivalent to the time complexity of computing the volume of polynomially many $k$ dimensional polytopes.
%

\subsection{Regret Analysis}
In this section we analyze the regret bound of Algorithm \ref{main-alg}.
In the main theorem of this section, Theorem \ref{regretb}, we show Algorithm \ref{main-alg}'s predictions for the first $T$ rounds has a regret bound of $O(k^4 (f(T)+1)\ln(T))$.
Note that after each action by the agent, the set of consistent values should satisfy the following new  constraints.
$$ \forall j\not= a_t , v_{a_t}-v_j+r_t \geq c_{a_t}-c_j $$
Lemma \ref{costlb} will be used later in the proof of Theorem \ref{regretb} to show that each time the prediction of the algorithm is wrong (meaning  $a_t\not=p_t$) the set $CV(t)$ shrinks. Before stating the lemma, we need to define the following notations:
$$U_{ij}(t)=max_{v \in CV(t)} \{v_i-v_j\}$$
$$L_{ij}(t)=min_{v \in CV(t)} \{v_i-v_j\}$$
\begin{lemma}\label{costlb}
If the predicted arm $p_t$ is not the arm $a_t$ that is played by the agent, then 
$$c^t_{a_t}-c^t_{p_t}\geq L_{a_tp_t}(t)+ \frac{1}{8k} (U_{a_tp_t}(t)-L_{a_tp_t}(t)).$$
\end{lemma}
\begin{proof}
Let us simplify the notations by omitting some of the indices: $a=a_t$, $p=p_t$, $L=L_{a_tp_t}$, $U=U_{a_tp_t}$, and $c=c_{a_t}-c_{p_t}$.
Suppose
\begin{equation}\label{lm1:eq1}
c < L+ \frac{1}{8k} (U-L) 
\end{equation}
for the sake of contradiction. 
Using Inequality~\eqref{lm1:eq1}, we show an arm $i$ exists such that its weight $w_i$ is higher than the weight of the arm $p$. 
Therefore, we have a contradiction because the algorithm chooses an arm $p$ such that $w_p=\max_{i\in[k]}~ w_i$. Lemma \ref{costlb} follows from this contradiction.

Let us define $G(z)=Pr_{v\sim Unif(CV(t))}[v_a-v_p<z]$ and $g(z)=\frac{dG(z)}{dz}$. We first show $g(z)$ is concave and non-negative in $[L_{ap},U_{ap}]$.

\begin{claim}
\label{gconcave}
$g(z)$ is concave and non-negative in $[L_{ap},U_{ap}]$.
\end{claim}
\begin{proof}
For simplicity and without loss of generality we suppose $CV(t)$ is full dimensional.
Following the definition, $G(z)$ is the probability that a randomly drawn point from $CV(t)$ is in the half space $v_a-v_p<z$.
In other words,	 $G(z)$ is ratio of the volume of intersection of $CV(t)$ and the half space $v_a-v_p<z$ over the volume of $CV(t)$, i.e., 
$$G(z)=\frac{\mbox{Vol}(CV(t)\cap \{v :~ v_a-v_p<z\})}{\mbox{Vol}(CV(t))}.$$ 
Now it is easy to see that the derivative of $G(z)$, $g(z)$, is the surface area of the intersection of the hyperplane $v_a-v_p=z$ and $CV(t)$. Therefore, the claim follows due to convexity of $CV(t)$.
\end{proof}

Considering Inequality \eqref{lm1:eq1}, the following claim proves an upper bound on the weight $w_p$ of arm $p$ and the next claim (Claim \ref{lm1:wi}) shows a lower bound on the sum of weights of all arms except arm $p$, i.e, $\sum_{i\not=p} w_i$. These claims will lead to the contradiction we need.

\begin{claim}\label{lm1:wp}
$w_p \leq 2 g(c)(c-L)$.
\end{claim}
\begin{proof}
Note that
\begin{equation}\label{lm1:eq2}
w_p \leq G(c)
\end{equation}
because we have $w_p=Pr_{v\sim Unif(CV(t))}[\forall j, v_p-c_p\geq v_j-c_j]$ and so
$$w_p \leq Pr_{v\sim Unif(CV(t))}[v_p-c_p\geq v_a-c_a]=G(c).$$
It suffices for the proof to show $g(x)\leq 2 g(c)$, $\forall x\in [L,c]$ because $G(c)=\int_{L}^{c} g(x) dx $. 
By Claim \ref{gconcave} we know that $g$ is a non-negative and concave function in $[L,U]$. Therefore, we have 
\[
\begin{array}{ll}
\forall x\in [L,c], ~~~ g(x) & \leq g(c)-\gamma(c-x)
\end{array}
\]
where $\gamma$ is the derivative of $g$ at point $c$. By concavity of $g$, we have
$\gamma \geq \frac{g(U)-g(c)}{U-c}.$ Therefore, for every $x\in [L,c]$, we have
\begin{eqnarray*}
g(x) & \leq& g(c)-\frac{g(U)-g(c)}{U-c}(c-x)\\
& \leq& g(c)+g(c)\cdot\frac{c-x}{U-c}\\
& \leq& 2g(c) 
\end{eqnarray*}
where the second inequality follows from the non-negativity of $g(U)$, and the last inequality holds because by Inequality \eqref{lm1:eq1}, $c-L \leq U - c$, and therefore for every $x\in [L,c]$, $\frac{c-x}{U-c}\leq 1$.
\end{proof}

\begin{claim}\label{lm1:wi}
$\sum_{i: i\not= p} w_i \geq  \frac{g(c)}{2}(U-c).$
 \end{claim}
\begin{proof}
Note that $\sum_i w_i=1$. Therefore, by Inequality \eqref{lm1:eq2}, we have
\begin{equation}
\label{eqn:clm3}
\sum_{i: i \not= p } w_i = 1-w_p  \geq 1-G(c) = G(U)-G(c).
\end{equation}
Since $g$ is a non-negative concave function on $[L,U]$, we have
$$\forall x\in [c,U], ~ g(x) \geq g(c) + \frac{g(U)-g(c)}{U-c}(x-c)$$
Therefore,
\begin{eqnarray*}
G(U)-G(c)  &=& \int_{c}^{U} g(x) dx \\
&\geq&  \int_{c}^{U} \left(g(c)+\frac{g(U)-g(c)}{U-c}\right)(x-c)dx\\
&=& \frac{g(c)+g(U)}{2}(U-c)\\
&\ge&\frac{g(c)}{2}(U-c).
\end{eqnarray*}

This, together with Inequality~\eqref{eqn:clm3} complete the proof of Claim~\ref{lm1:wi}.
\end{proof}

Now we show a contradiction using Claim \ref{lm1:wp}, Claim \ref{lm1:wi} and Equation \eqref{lm1:eq1}. 
Note that $g(c)>0$ and $U-c>\frac{U-L}{2}$ by Claim \ref{lm1:wp} and Inequality \eqref{lm1:eq1}, respectively. Therefore,
$$
w_p  
\leq   2 g(c)(c-L)
\leq  \frac{g(c)}{4k} (U-L)
< \frac{g(c)}{2k} (U-c),
$$
where the first and the second inequalities follow from Claim \ref{lm1:wp} and Inequality \eqref{lm1:eq1}, respectively.
On the other hand, using Claim \ref{lm1:wi} we know there exists an arm $i$ such that
$$ w_i \geq\frac{g(c)}{2k}(U-c).$$
Therefore, we have 
$w_i\geq \frac{g(c)}{2k}(U-c) > w_p$ which contradicts the way $p$ is selected by Algorithm~\ref{main-alg}.
\end{proof}

\begin{theorem}\label{regretb}
Total regret of Algorithm \ref{main-alg} for the first $T$ rounds is bounded by $O(k^4(f(T)+1)\ln(T))$.
\end{theorem}
\begin{proof}
To prove the theorem, we show that 
\begin{equation}
\label{eqn:goal}
\sum_{t\leq T } pr_t \leq f(T)+k^2\lambda H(T)(f(T)+1)
\end{equation}
for $\lambda >2+\frac{1}{1-\delta(1 -\ln(\delta))}$ and $\delta=1-\frac{1}{8k}$. Here $H(T)$ denotes the harmonic series.
Let $v^*$ denote the actual value vector of the arms. By the definition of regret we have
\begin{eqnarray*}
pr_t 
& = &(v^*_{o_t}-c^t_{o_t}) - ( v^*_{p_t} - c^t _{p_t}) \\
& = &((v^*_{o_t}-c^t_{o_t}) - (v^*_{a_t}-c^t_{a_t})) + ((v^*_{a_t}-c^t_{a_t}) - ( v^*_{p_t} - c^t _{p_t}))\\ 
& = &ar_t + ((v^*_{a_t}-c^t_{a_t}) - ( v^*_{p_t} - c^t _{p_t}))
\end{eqnarray*}
Let us define $er_t=\max(0,(v^*_{a_t}-c^t_{a_t}) - ( v^*_{p_t} - c^t _{p_t}))$.
Therefore, 
$$\sum_{t\leq T} pr_t \leq \sum_{t\leq T} ar_t + \sum_{t\leq T} er_t \leq f(T)+\sum_{t\leq T} er_t.$$ 
Therefore, to prove Inequality~\eqref{eqn:goal}, it is enough to show $ \sum_{t \leq T} er_t\leq k^2\lambda H(T)(f(T)+1)$.
We define $B_{\alpha\beta}(T)=\{t :~ t\leq T \textrm{ and } (a_t,p_t)=(\alpha,\beta)\}$. Note that we have
\begin{eqnarray}
\label{sumer}
\sum_{t\leq T} er_t & = &
\sum_{\alpha,\beta} \sum_{t\in B_{\alpha\beta}(T)} er_t\nonumber\\
& \leq& k^2  \cdot\max_{\alpha,\beta} \{\sum_{t\in B_{\alpha\beta}(T)} er_t\}.
\end{eqnarray}
Therefore, to prove Inequality~\eqref{eqn:goal}, it is enough to  show that for every $\alpha, \beta$,
\[
\sum_{t\in B_{\alpha\beta}(T)} er_t \leq \lambda H(T)(f(T)+1).
\]
Let us fix $\alpha$ and $\beta$. 
Suppose $l=|B_{\alpha\beta}(T)|$ and $B_{\alpha\beta}(T)=\{t_1,\ldots,t_{l}\}$ where $t_1 < \cdots < t_{l}$.
We only consider cases where $\alpha\not=\beta$ because 
$\forall \alpha, \sum_{t\in B_{\alpha\alpha}}er_t=0$.
Therefore, using Lemma \ref{costlb} we know $L(t_i) \leq c^{t_i}_\alpha-c^{t_i}_\beta $. That gives
\[
\begin{array}{lll}
er_{t_i} & = \max(0, (v^*_{\alpha}-v^*_{\beta})-(c^{t_i}_{\alpha}-c^{t_i}_{\beta}))\\ 
& \leq \max(0, (v^*_{\alpha}-v^*_{\beta}) - L(t_i)  )
\end{array}
\]
In following claim we show $(v^*_{\alpha}-v^*_{\beta}) - L(t_i)$ is bounded by $\frac{\lambda (f(t_i)+1)}{i}$.
\begin{claim}\label{lm2:clm1}
For every $t_i\in B_{\alpha\beta}(T)$, we have
$$(v_\alpha^* - v^*_{\beta}) - L(t_i) \leq \frac{\lambda (f(t_i)+1)}{i}.$$
\end{claim}
\begin{proof}
The proof is by contradiction. Suppose there is a $t_i$ such that 
\begin{equation}\label{hypo1}
(v^*_\alpha-v^*_\beta) - L(t_i)>\frac{\lambda (f(t_i)+1)}{i}.
\end{equation}
Let $t_i$ be the smallest such $t_i$. Therefore,
\begin{equation}\label{hypo2}
\forall j < i, \ (v_\alpha^* - v^*_{\beta})- L(t_j) \leq \frac{\lambda (f(t_j)+1)}{j}.
\end{equation}
Let $\hat{v}\in CV(t_i)$ be a point that minimizes $v_\alpha-v_\beta$, i.e., $\hat{v}_\alpha-\hat{v}_\beta=L(t_i)$.
Note that we have $i>1$ because the values are bounded by 1.
Let us recall the definition of $CV(t_i)$ here. A vector $v$ is in 
$CV(t_i)$ if $\exists r\in \mathbb{R}^T$ such that:
\[
\begin{array}{ll}
 \forall t \in [t_i-1]\forall j &:~v_{a_t}-c^t_{a_t} \geq v_{j}-c^t_j-r_t \\
\forall t \in [t_i-1]&:	   ~  \sum_{h\leq t} r_h \leq f(t) 
	    
\end{array}
\]
This can be written as:
\[
\begin{array}{ll}
\forall t \in [t_i-1]\forall j &:~ (v_{j}-c^t_j)-(v_{a_t}-c^t_{a_t}) \leq r_t \\
\forall t \in [t_i-1]&:~ \sum_{h\leq t} r_h \leq f(t)
\end{array}
\]
Since $\hat{v}\in CV(t_i)$, we have
$$ \sum_{t< t_i} \max(0,(\hat{v}_{p_t} - c^{t}_{p_t})- (\hat{v}_{a_t} - c^{t}_{a_t}))\leq \sum_{t< t_i} r_t \leq  f(t_i-1) $$
Note that $B_{\alpha\beta}(t_i-1)\subset [t_i-1]$. Therefore, we get
$$\sum_{t_j\in B_{\alpha\beta}(t_i-1)}  \max(0,(\hat{v}_{\beta} - c^{t_j}_{\beta})- (\hat{v}_{\alpha} - c^{t_j}_{\alpha}) ) \leq f(t_i-1).$$
Note that we can write 
$(\hat{v}_{\beta} - c^{t_j}_{\beta})- (\hat{v}_{\alpha} - c^{t_j}_{\alpha})$ as
 $$ ((v^*_\alpha - v^*_\beta) - L(t_i))  -( (v^*_\alpha - v^*_\beta) - (c^{t_j}_\alpha - c^{t_j}_\beta) )$$
because $\hat{v}_\alpha-\hat{v}_\beta=L(t_i)$. 
If we combine the above equations we get

\begin{eqnarray}
\label{eqn:9}
f(t_i-1) 
& \geq&
\sum_{j  <i} \max(0,((v^*_\alpha - v^*_\beta) - L(t_i))  -( (v^*_\alpha - v^*_\beta) - (c^{t_j}_\alpha - c^{t_j}_\beta) ))\nonumber\\
& \geq&
\sum_{j < i} \max(0,\frac{\lambda (f(t_i)+1)}{i} -( (v^*_\alpha - v^*_\beta) - (c^{t_j}_\alpha - c^{t_j}_\beta) )) 
\end{eqnarray}
where the second inequality follows from Inequality \eqref{hypo1}.
On the other hand, we have
\begin{eqnarray}
\label{eqn:10}
(v^*_\alpha - v^*_\beta) - (c^{t_j}_\alpha - c^{t_j}_\beta) 
&\leq&
(v^*_\alpha - v^*_\beta) -\left( (1-\frac{1}{8k}) L(t_j)+ \frac{1}{8k} U(t_j)\right) \nonumber\\
& \leq& 
(1-\frac{1}{8k}) ( (v^*_\alpha - v^*_\beta) -L(t)),
\end{eqnarray}
where the first inequality follows from Lemma~\ref{costlb} and the second inequality follows from the fact that $U(t_j)\geq v^*_\alpha - v^*_\beta$.
Inequalities~\eqref{eqn:9} and \eqref{eqn:10} imply:

\begin{eqnarray}
\label{lm1:f}
f(t_i-1) 
& \geq&
\sum_{j <i} \max(0,\frac{\lambda (f(t_i)+1)}{i} -((1-\frac{1}{8k}) ( (v^*_\alpha - v^*_\beta) -L(t_j)))
\end{eqnarray}

%
%
%
Recall $\delta=1-\frac{1}{8k}$. If we apply Equation \eqref{hypo2} into Equation \eqref{lm1:f} we get:

\begin{eqnarray*}
f(t_i-1) 
& \geq&
\sum_{j < i} \max\left(0,\frac{\lambda (f(t_i)+1)}{i} -\delta\lambda \frac{f(t_j)+1}{j}\right)\\
& \geq&
\sum_{\lfloor \delta i \rfloor\leq j < i} \max\left(0,\frac{\lambda (f(t_i)+1)}{i} -\delta  \lambda \frac{f(t_j)+1}{j}\right)\\
& \geq&
\sum_{\lfloor \delta i \rfloor\leq j < i} \frac{\lambda (f(t_i)+1)}{i} -\delta\lambda \frac{(f(t_j)+1)}{j}\\
& \geq&
\sum_{\lfloor \delta i \rfloor\leq j < i} \lambda (f(t_i)+1) (\frac{1}{i} -\delta \frac{1}{j}),
\end{eqnarray*}
where the last inequality follows from the fact that $f$ is monotone and increasing.
With some straightforward calculations on the above we get:
$$ 1 \geq \lambda(1-\delta(1 +\sum_{\lfloor \delta i \rfloor \leq j<i} \frac{1}{j}))$$
It is easy to see $\sum_{\lfloor \delta i \rfloor \leq j< i} \frac{1}{j} \leq \ln(\frac{1}{\delta})$ since $\frac{7}{8}\leq \delta < 1$. Therefore,
$$ 1 \geq \lambda((1-\delta) - \delta \ln(\frac{1}{\delta}))$$
which is a contradiction because $\lambda> \frac{1}{(1-\delta) - \delta \ln(\frac{1}{\delta})}$ and $(1-\delta) - \delta \ln(\frac{1}{\delta})>0$. The claim follows from this contradiction.
\end{proof}
By Claim~\ref{lm2:clm1},
\begin{eqnarray*}
\sum_{t_i \in B_{\alpha\beta}(T)} er_{t_i} 
& \leq&
\sum_{t_i \in B_{\alpha\beta}(T)} \frac{\lambda (f(t_i)+1)}{i}\\
& \leq&
\lambda (f(T)+1) \sum_{i\leq |l|} \frac{1}{i} \\
& \leq&
\lambda (f(T)+1)H(l) \leq \lambda (f(T)+1)H(T),
\end{eqnarray*}
which completes the proof of Theorem~\ref{regretb}.
\end{proof}

\subsection{Bounding the number of wrong predictions}
Note that predicting the exact arm an advertiser would choose is not always feasible. 
If there is more than one optimal arm, finding which one the advertiser would choose is not possible. 
Therefore, we need an assumption that the optimal arm is unique in every time step.

The following theorem is a corollary of Theorem \ref{regretb}. It bounds the number of wrong predictions of Algorithm 1. In this theorem, the utility of an arm is defined as the value of the arm minus the cost of playing it.
\begin{theorem}
If the utility of the optimal arm is higher than the utility of other arms by $\delta$ for every time step, then the number of mistakes is bounded by $\frac{k^4(f(T)+1)\log(T)+f(T)}{\delta}$.
\end{theorem}
\begin{proof}
Let $m_o(t)$ be the number of wrong predictions in which the algorithm chooses the optimal arm, i.e., $p_t=o_t$. Note that in such time steps the agent has a regret of at least $\delta$. Therefore, the overall regret of the agent is lower bounded by $m_o(t)\delta$, and so $m_o(t)\leq \frac{f(t)}{ \delta }$.

Let $m_a(t)$ be the number of wrong predictions in which the algorithm does not choose the optimal arm. In such time steps the algorithm has a regret of at least $\delta$. Therefore, the overall regret of the algorithm is at least $m_a(t)\delta$. Using Theorem \ref{regretb}, we get 
$$m_a(t) \leq \frac{k^4(f(T)+1)\log(T)}{\delta}$$.

The total number of wrong predictions up to time step $t$ is $m_o(t)+m_a(t)\leq \frac{f(t)}{ \delta } + \frac{k^4(f(T)+1)\log(T)}{\delta}$.
\end{proof}

\section{Lower bound}

In this section, we show a lower bound on the prediction regret that holds even when the regret of the agent is zero, that is, $f(T)=0$.
We prove that there is no algorithm that can predict the agent's actions with a regret bound lower than $\frac{k}{4}$, even when $f(T)=0$.

\begin{theorem}
Given any algorithm $\mathcal{A}$, there exists a sequence of costs in which we have $\sum_{t\leq k/2} pr_t \geq \frac{k}{4}$.
\end{theorem}
\begin{proof}
For simplicity suppose $k$ is even.
Consider the following sequence of cost vectors. 
\[
\begin{array}{ll}
c^1& =(0,0,H,H,\ldots,H,H)\\
c^2& =(H,H,0,0,H,\ldots,H)\\
\vdots\\
c^{k/2}& =(H,H,\ldots,H,H,0,0)
\end{array}
\]
where $H$ is any constant bigger than 1.
Formally, $c^t=(c^t_1,\ldots,c^t_k)$ where
\begin{equation} \label{cd}
c^t_i=
 \begin{cases} 
      0 & i\in \{2t,2t-1\} \\
      H & \text{otherwise}
   \end{cases}
\end{equation}
Note that at each time step $t$, the algorithm has no information about arms $2t$ and $2t-1$.
Therefore, the algorithm cannot do better than choosing at random. 
If we set the rewards for arms as follows 
\begin{equation} \label{vs}
v^*_i=
 \begin{cases} 
      1 & \text{$i$ is even} \\
      0 & \text{$i$ is odd}
   \end{cases}
\end{equation}
then the algorithm has a regret of $\frac{1}{2}$ at every step. Therefore, the total regret will be at least $\frac{k}{4}$
\end{proof}

\section{Conclusion}

In this paper, we studied a multi-armed bandits setting where in each step, a cost for playing each arm is announced to the agent. We proved that if we observe an agent that achieves a regret of at most $f(T)$, then even without observing any rewards, we can learn to play with a regret of at most $O(k^4 (f(T)+1) \log(T))$, where $k$ is the number of arms.

We used this model to capture applications like ad auctions, where the goal is to understand and predict the behavior of an advertiser with unknown utility and unobserved rewards. 

There are several problems that are left open. The most natural open question is to find the best regret bound achievable in our setting. The only lower bound we know is $O(k)$ in the case that $f(T)=0$. Also, the broader question of predicting an selfish agent's actions in a dynamic environment without observing her rewards is open in more complicated settings.

\bibliographystyle{splncs04}
\bibliography{Volume_Algorithm}

\begin{thebibliography}{10}
\providecommand{\url}[1]{\texttt{#1}}
\providecommand{\urlprefix}{URL }
\providecommand{\doi}[1]{https://doi.org/#1}

\bibitem{SD10}
Athey, S., Nekipelov, D.: {A structural model of sponsored search advertising
  auctions}. Sixth ad auctions workshop  (2010)

\bibitem{AH07}
Athey, S., Haile, P.A.: {Nonparametric Approaches to Auctions}. In: {Handbook
  of Econometrics}. Handbook of Econometrics, Elsevier (2007)

\bibitem{backstrom2011network}
Backstrom, L., Kleinberg, J.: Network bucket testing. In: Proceedings of the
  20th international conference on World wide web. pp. 615--624. ACM (2011)

\bibitem{Braverman}
Braverman, M., Mao, J., Schneider, J., Weinberg, M.: Selling to a no-regret
  buyer. In: Proceedings of the 2018 ACM Conference on Economics and
  Computation. pp. 523--538 (2018)

\bibitem{BG11}
Broder, A., Gabrilovich, E., Josifovski, V., Mavromatis, G., Smola, A.: Bid
  generation for advanced match in sponsored search. In: Proceedings of the
  Fourth ACM International Conference on Web Search and Data Mining. pp.
  515--524. New York, NY, USA (2011)

\bibitem{BC12}
Bubeck, S., Cesa{-}Bianchi, N.: Regret analysis of stochastic and nonstochastic
  multi-armed bandit problems. CoRR  \textbf{abs/1204.5721} (2012)

\bibitem{CGPV03}
Campo, S., Guerre, E., Perrigne, I., Vuong, Q.: {Semiparametric Estimation of
  First-price Auctions with Risk Averse Bidders}. Working papers, Centre de
  Recherche en Economie et Statistique (2003)

\bibitem{CDE07}
Cary, M., Das, A., Edelman, B., Giotis, I., Heimerl, K., Karlin, A.R., Mathieu,
  C., Schwarz, M.: Greedy bidding strategies for keyword auctions. In:
  Proceedings of the 8th ACM Conference on Electronic Commerce. pp. 262--271.
  EC '07, ACM, New York, NY, USA (2007). \doi{10.1145/1250910.1250949}

\bibitem{chawla2014mechanism}
Chawla, S., Hartline, J., Nekipelov, D.: Mechanism design for data science. In:
  Proceedings of the fifteenth ACM conference on Economics and computation. pp.
  711--712. ACM (2014)

\bibitem{FKT10}
Feldman, M., Kalai, A., Tennenholtz, M.: Playing games without observing
  payoffs. In: ICS. pp. 106--110 (2010)

\bibitem{GPV00}
Guerre, E., Perrigne, I., Vuong, Q.: {Optimal Nonparametric Estimation of
  First-Price Auctions}. Econometrica  \textbf{68}(3),  525--574 (May 2000)

\bibitem{GPV09}
Guerre, E., Perrigne, I., Vuong, Q.: {Nonparametric Identification of Risk
  Aversion in First-Price Auctions Under Exclusion Restrictions}. Econometrica
  (2009)

\bibitem{BP03}
Jofre-Bonet, M., Pesendorfer, M.: {Estimation of a Dynamic Auction Game}.
  Econometrica  (2003)

\bibitem{KS15}
Kuleshov, V., Schrijvers, O.: Inverse game theory: Learning utilities in
  succinct games. In: WINE (2015)

\bibitem{lang2007finding}
Lang, K.J., Andersen, R.: Finding dense and isolated submarkets in a sponsored
  search spending graph. In: Proceedings of the sixteenth ACM conference on
  Conference on information and knowledge management. pp. 613--622. ACM (2007)

\bibitem{mizuta2000agent}
Mizuta, H., Steiglitz, K.: Agent-based simulation of dynamic online auctions.
  In: Simulation Conference, 2000. Proceedings. Winter. vol.~2, pp. 1772--1777.
  IEEE (2000)

\bibitem{NST15}
Nekipelov, D., Syrgkanis, V., Tardos, E.: Econometrics for learning agents. In:
  Proceedings of the Sixteenth ACM Conference on Economics and Computation. pp.
  1--18. ACM (2015)

\bibitem{PK11}
Pin, F., Key, P.: Stochastic variability in sponsored search auctions:
  Observations and models. In: Proceedings of the 12th ACM Conference on
  Electronic Commerce. pp. 61--70 (2011)

\bibitem{tang2010overlapping}
Tang, D., Agarwal, A., O'Brien, D., Meyer, M.: Overlapping experiment
  infrastructure: More, better, faster experimentation. In: Proceedings of the
  16th ACM SIGKDD international conference on Knowledge discovery and data
  mining. pp. 17--26. ACM (2010)

\bibitem{V07}
Varian, H.R.: {Position auctions}. International Journal of Industrial
  Organization  (2007)

\bibitem{XG13}
Xu, H., Gao, B., Yang, D., Liu, T.Y.: Predicting advertiser bidding behaviors
  in sponsored search by rationality modeling. In: Proceedings of the 22nd
  international conference on World Wide Web (May 2013)

\end{thebibliography}

\end{document}